\documentclass[12 pt]{amsart}
\usepackage{amssymb,latexsym,amsmath,amscd,amsthm,graphicx, color}
\usepackage[all]{xy}
\usepackage{pgf,tikz}
\usepackage{mathrsfs}
\usepackage{cite}
\usetikzlibrary{arrows}
\definecolor{uuuuuu}{rgb}{0.26666666666666666,0.26666666666666666,0.26666666666666666}
\definecolor{xdxdff}{rgb}{0.49019607843137253,0.49019607843137253,1.}
\definecolor{ffqqqq}{rgb}{1.,0.,0.}

\definecolor{uuuuuu}{rgb}{0.26666666666666666,0.26666666666666666,0.26666666666666666}
\definecolor{qqwuqq}{rgb}{0.,0.39215686274509803,0.}
\definecolor{zzttqq}{rgb}{0.6,0.2,0.}
\definecolor{xdxdff}{rgb}{0.49019607843137253,0.49019607843137253,1.}
\definecolor{qqqqff}{rgb}{0.,0.,1.}
\definecolor{cqcqcq}{rgb}{0.7529411764705882,0.7529411764705882,0.7529411764705882}

\theoremstyle{plain}

\newtheorem{cor}[subsection]{Corollary}
\newtheorem{conj}[subsection]{Conjecture}

\newtheorem{theorem}[subsection]{Theorem}

\newtheorem{lemma}[subsection]{Lemma}

\theoremstyle{definition}
\newtheorem{prop}[subsection]{Proposition}

\newtheorem{note}[subsection]{Note}

\newcommand{\sci}{\subset}
\newcommand{\set}[1]{\{#1\}}

\newcommand{\ga}{\alpha}
\newcommand{\gb}{\beta}

\newcommand{\gd}{\delta}
\renewcommand{\gg}{\gamma}

\newcommand{\gl}{\lambda}

\newcommand{\tbf}{\textbf}
\newcommand{\tit}{\textit}

\newcommand{\D}[1]{\mathbb{#1}}

\newcommand{\te}{\text}

\newcommand{\tri}{\triangle}
\begin{document}

\tbf{To appear, Real Analysis Exchange}
\title{Quantization for uniform distributions on equilateral triangles}

\author{Carl P. Dettmann}
\address{University of Bristol\\
School of Mathematics\\
University Walk\\
Bristol BS8 1TW\\
UK.}
\email{Carl.Dettmann@bris.ac.uk}

\author{ Mrinal Kanti Roychowdhury}
\address{School of Mathematical and Statistical Sciences\\
University of Texas Rio Grande Valley\\
1201 West University Drive\\
Edinburg, TX 78539-2999, USA.}
\email{mrinal.roychowdhury@utrgv.edu}

\thanks{The research of the first author was supported by the Engineering and Physical Sciences Research Council (EPSRC) Grant grant EP/N002458/1, and that of the second author was supported by U.S. National Security Agency (NSA) Grant H98230-14-1-0320}

\date{}
\maketitle

\pagestyle{myheadings}\markboth{Carl P. Dettmann and Mrinal Kanti Roychowdhury}{Quantization for uniform distributions on equilateral triangles}

\begin{abstract}
We approximate the uniform measure on an equilateral triangle by a measure supported on $n$ points.  We find the optimal sets of points ($n$-means) and
corresponding approximation (quantization) error for $n\leq4$, give numerical optimization results for $n\leq 21$, and a bound on the quantization error for
$n\to\infty$.  The equilateral triangle has particularly efficient quantizations due to its connection with the triangular lattice.  Our methods can be applied to the uniform distributions on general sets with piecewise smooth boundaries.
\end{abstract}

\section{Introduction}
The representation of a given quantity with less information is often referred to as `quantization' and it is an important subject in information theory. It has broad applications in signal processing, telecommunications, data compression, image processing and cluster analysis. We refer to \cite{GG, GN, Z} for surveys on the subject and comprehensive lists of references to the literature, see also \cite{GKL}.  Rigorous mathematical treatment of the quantization theory is given in Graf-Luschgy's book (see \cite{GL1}).

Let $P$ denote a Borel probability measure on $\D R^d$ and let $\|\cdot\|$ denote the Euclidean norm on $\D R^d$ for any $d\geq 1$. We consider an approximation of $P$ by a
measure supported on at most a finite number of points, $n$.  The $n$th \textit{quantization
error} for $P$ is defined by
\begin{equation*} \label{eq1} V_n:=V_n(P)=\inf \Big\{\int \min_{a\in\alpha} \|x-a\|^2 dP(x) : \alpha \subset \mathbb R^d, \text{ card}(\alpha) \leq n \Big\},\end{equation*}
where the infimum is taken over all subsets $\alpha$ of $\mathbb R^d$ with card$(\alpha)\leq n$ for $n\geq 1$. Notice that if $\int \| x\|^2 dP(x)<\infty$, then there is some set $\alpha$ for
which the infimum is achieved (see \cite{GL1}). This set $\ga$ can then be used to give a best approximation of
$P$ by a discrete probability supported on a set with no more than
$n$ points.  Such a set $\ga$ for which the infimum occurs and contains no more than $n$ points is called an \tit{optimal set of $n$-means}, or \tit{optimal set of $n$-quantizers}. It is known that for a continuous probability measure $P$ an optimal set of $n$-means always has exactly $n$ elements (see \cite{GL1}). The probability measure $P$ considered in this paper is a uniform distribution which is absolutely continuous with respect to the Lebesgue measure $\gl$, and so there exists a probability density function $f$, known as Radon-Nikodym derivative of $P$ with respect to $\gl$, with $f\geq 0$ and $\int f d\gl=1$ such that for any Borel subset $B\sci \D R^d$, we have
\begin{equation} \label{eq001} P(B) =\int_B f d\gl.\end{equation}
Given a finite subset $\ga\sci \D R^d$, the \tit{Voronoi region} generated by $a\in \ga$ is defined by
\[M(a|\ga)=\set{x \in \D R^d : \|x-a\|=\min_{b \in \ga}\|x-b\|}\]
i.e., the Voronoi region generated by $a\in \ga$ is the set of all points in $\D R^d$ which are closest to $a \in \ga$, and the set $\set{M(a|\ga) : a \in \ga}$ is called the \tit{Voronoi diagram} or \tit{Voronoi tessellation} of $\ga$. A Borel measurable partition $\set{A_a : a \in \ga}$ of $\D R^d$  is called a \tit{Voronoi partition} of $\D R^d$ with respect to $\ga$ (and $P$) if $P$-almost surely, we have
\[A_a \sci M(a|\ga) \te{ for every $a \in \ga$}.\]
Notice that if $\ga=\set{a_1, a_2, \cdots, a_n}$ is an optimal set of $n$-means for $P$ and $\set{A_1, A_2, \cdots, A_n}$ is a Voronoi partition with respect to $\ga$, then
\[V_n=\sum_{i=1}^n \int_{A_i} \|x-a_i\|^2 dP(x).\]

Let us now state the following proposition (see \cite{GG, GL1}).
\begin{prop} \label{prop10}
Let $\alpha$ be an optimal set of $n$-means, $a \in \alpha$, and $M (a|\ga)$ be the Voronoi region generated by $a\in \ga$, i.e.,
\[M(a|\ga)=\{x \in \mathbb R^d : \|x-a\|=\min_{b \in \alpha} \|x-b\|\}.\]
Then, for every $a \in\alpha$,

$(i)$ $P(M(a|\ga))>0$, $(ii)$ $ P(\partial M(a|\ga))=0$, $(iii)$ $a=E(X : X \in M(a|\ga))$, and $(iv)$ $P$-almost surely the set $\set{M(a|\ga) : a \in \ga}$ forms a Voronoi partition of $\D R^d$.
\end{prop}

Let $\alpha$ be an optimal set of $n$-means and  $a \in \alpha$, then by Proposition~\ref{prop10}, we have
\begin{align*}
a=\frac{1}{P(M(a|\ga))}\int_{M(a|\ga)} x dP=\frac{\int_{M(a|\ga)} x dP}{\int_{M(a|\ga)} dP}=\frac{\int_{M(a|\ga)} x f(x) d\gl}{\int_{M(a|\ga)} f(x) d\gl},
\end{align*}
which implies that $a$ is the centroid of the Voronoi region $M(a|\ga)$ associated with the probability measure $P$ (see also \cite{DFG}).

The classical Cantor set $C$ is generated by the two contractive similarity mappings $S_1(x)=\frac 1 3 x$ and $S_2(x)=\frac 1 3 x +\frac 23 $ for all $x \in \D R$. Then, there exists a unique Borel probability measure $P$ on $\D R$ with support $C$ such that $P=\frac 1 2 P\circ S_1^{-1}+\frac 1 2 P\circ S_2^{-1}$, where $P\circ S_i^{-1}$ denotes the image measure of $P$ with respect to $S_i$ for $i=1, 2$ (see \cite{H}). Such a probability measure is mutually singular with respect to the Lebesgue measure, and in \cite{GL2}, Graf-Luschgy investigated the optimal quantization for this measure $P$.

In this paper, we have considered a uniform distribution on an equilateral triangle, and investigated the optimal sets of $n$-means and the $n$th quantization error for this distribution for all $n\geq 1$. Moreover, in Theorem~\ref{th2}, we have shown that the  Voronoi regions generated by the two points in an optimal set of two-means partition the equilateral triangle into an isosceles trapezoid and an equilateral triangle in the Golden ratio. In subsequent sections, we find the optimal sets of three- and four-means.  In the last section, in Theorem~\ref{th:6.3} and in its corollary, we have given some numerical optimization results and conjectures about the optimal configurations for $n$ points, a rigorous bound on the quantization error for $n\to\infty$, and a final conjecture about uniform distributions in more general geometries.

Our approach illustrates methods for far more general geometries, including the use of symmetry to find optimal sets for small $n$, numerical optimisation for intermediate $n$, and configurations close to the triangular lattice for large $n$.  Efficient quantization due to matching of the boundaries to a triangular lattice is only possible in polygons with all angles a multiple of $\pi/3$.  The simplest and most natural example of this is the equilateral triangle.

\section{Some basic results relating to quantization and uniform distributions}
In this section we give some basic results relating to optimal sets and the uniform probability distributions defined on equilateral triangles.
Let $X=(X_1, X_2)$ be a bivariate continuous random variable with uniform distribution taking values on the triangle $\tri$ with vertices $(0, 0), \, (1, 0), \, (\frac 1 2, \frac{\sqrt 3}{2})$. Then, the  probability density function (pdf) $f(x_1, x_2)$  of the random variable $X$ is given by
\[f(x_1, x_2)=\left\{\begin{array}{ccc}
\frac{4}{\sqrt 3} & \te{ for }   0<x_1<\frac 12, \  0<x_2<\sqrt 3 x_1, \\
\frac{4}{\sqrt 3} &  \te{ for }   \frac 1 2<x_1<1, \ 0<x_2<-\sqrt 3 x_1+\sqrt 3, \\
\ 0  & \te{ otherwise}.
\end{array}\right.
\]
Notice that the pdf satisfies the following two necessary conditions:

$(i)$ $f(x_1, x_2)\geq 0$ for all $(x_1, x_2) \in \D R^2$, \te{ and }

$(ii)$ $\iint_{\D R^2} f(x_1, x_2)\, dx_1 dx_2=\int_0^{\frac 12}\int_0^{\sqrt 3 x_1}f(x_1, x_2)\,dx_2 dx_1+\int_{\frac 12}^{1}\int_0^{-\sqrt 3 x_1+\sqrt 3}f(x_1, x_2)\,dx_2 dx_1=1$.

Moreover, one should notice that the  pdf of the bivariate random variable $X$ can also be written in the following form:
\[f(x_1, x_2)=\left\{\begin{array}{ccc}
\frac{4}{\sqrt 3} & \te{ for }   0<x_2<\frac {\sqrt 3}{2},  \ \frac{x_2}{\sqrt 3}<x_1<\frac{\sqrt 3-x_2}{\sqrt 3}, \\
\ 0 &  \te{ otherwise}.
\end{array}\right.
\]
Let $f_1(x_1)$ and $f_2(x_2)$ represent the marginal pdfs of the random variables $X_1$ and $X_2$ respectively. Then, following the definitions in Probability Theory, we have
\[f_1(x_1)=\int_{-\infty}^\infty f(x_1, x_2) \,dx_2 \te{ and } f_2(x_2)=\int_{-\infty}^\infty f(x_1, x_2) \,dx_1.\]
Since $\int_0^{\sqrt 3 x_1} f(x_1, x_2)\,dx_2=4x_1$ for $0< x_1<\frac 12$, and $\int_0^{-\sqrt 3 x_1+\sqrt 3} f(x_1, x_2)\,dx_2=4(1-x_1)$ for $\frac 12<x_1<1$,
we have
\[f_1(x_1)=\left\{\begin{array}{ccc}
4 x_1 & \te{ for }   0<x_1<\frac 12, \\
4(1-x_1) & \te{ for }   \frac 1 2<x_1<1,\\
0   & \te{ otherwise}.
\end{array}\right.
\]
Similarly, we can write
\[f_2(x_2)=\left\{\begin{array}{cc} \frac {4}{\sqrt 3}(1-\frac{2x_2}{\sqrt 3}) & \te{ for } 0<x_2<\frac{\sqrt 3}{2},\\
0   & \te{ otherwise}.
\end{array}\right.
\]
Notice that both $f_1(x_1)$ and $f_2(x_2)$ satisfy the necessary conditions for pdfs: $f_1(x_1)\geq 0$, $f_2(x_2)\geq 0$ for all $x_1, x_2\in \D R$, and
\[\int_{-\infty}^\infty f_1(x_1)\,dx_1=1=\int_{-\infty}^\infty f_2(x_2)\,dx_2.\]
For a random variable $Y$, let $E(Y)$  and $V(Y)$ represent the expected vector and the expected squared distance of $Y$. Let $i$ and $j$ be the unit vectors in the positive directions of $x_1$ and $x_2$-axes respectively. By the position vector $\tilde a$ of a point $A$, it is meant that $\overrightarrow{OA}=\tilde a$. In the sequel, we will identify the position vector of a point $(a_1, a_2)$ by $(a_1, a_2):=a_1 i +a_2 j$, and apologize for any abuse in notation. For any two vectors $\vec u$ and $\vec v$, let $\vec u \cdot \vec v$ denote the dot product between the two vectors $\vec u$ and $\vec v$. Then, for any vector $\vec v$, by $(\vec v)^2$, we mean $(\vec v)^2:= \vec v\cdot \vec v$. Thus, $|\vec v|:=\sqrt{\vec v\cdot \vec v}$, which is called the length of the vector $\vec v$. For any two position vectors $\tilde a:=( a_1, a_2)$ and $\tilde b:=( b_1, b_2)$, we write $\rho(\tilde a, \tilde b):=(( a_1-b_1, a_2-b_2))^2=(a_1-b_1)^2 +(a_2-b_2)^2$.

Let us now prove the following lemma.
\begin{lemma}
Let $X=(X_1, X_2)$ be a bivariate continuous random variable with uniform distribution taking values on the triangle $\tri$. Then,
\[E(X)=(E(X_1), E(X_2))=(\frac 1 2, \frac{\sqrt 3}{6}) \te{ and } V(X)=V(X_1)+V(X_2)=\frac 1 {12}.\]
\end{lemma}
\begin{proof}
We have
\begin{align*}
E(X_1)& =\int_{-\infty}^\infty x_1 f_1(x_1)\, dx_1=\int_0^{\frac{1}{2}} 4 x_1^2 \, dx_1+\int_{\frac{1}{2}}^1 4 \left(1-x_1\right) x_1 \, dx_1=\frac{1}{2},\\
E(X_2) &=\int_{-\infty}^\infty x_2 f_2(x_2)\, dx_2=\int_0^{\frac{\sqrt{3}}{2}} \frac{4}{\sqrt{3}} \left(1-\frac{2 x_2}{\sqrt{3}}\right) x_2 \, dx_2=\frac {\sqrt 3}{6},\\
E(X_1^2)& =\int_{-\infty}^\infty x_1^2 f_1(x_1)\, dx_1=\int_0^{\frac{1}{2}} 4 x_1^3 \, dx_1+\int_{\frac{1}{2}}^1 4 \left(1-x_1\right) x_1^2 \, dx_1=\frac{7}{24},\\
E(X_2^2)&=\int_{-\infty}^\infty x_2^2 f_2(x_2)\, dx_2=\int_0^{\frac{\sqrt{3}}{2}} \frac{4}{\sqrt{3}} \left(1-\frac{2 x_2}{\sqrt{3}}\right) x_2^2 \, dx_2=\frac {1}{8},
\end{align*}
and so, \begin{align*}
E(X)& =\iint(x_1 i+x_2 j) f(x_1, x_2) dx_1dx_2 =i \int x_1 f_1(x_1) dx_1+j \int x_2 f_2(x_2) dx_2\\
& =(E(X_1), E(X_2))=(\frac 1 2, \frac{\sqrt 3}{6}),\\
V(X_1)& =E(X_1^2)-[E(X_1)]^2=\frac{1}{24} \te{ and } V(X_2)=E(X_2^2)-[E(X_2)]^2=\frac{1}{24}.
\end{align*}
Thus, we have
\[V(X)=E\|X-E(X)\|^2=\iint \Big((x_1-E(X_1))^2 +(x_2-E(X_2))^2\Big)f(x_1, x_2)\, dx_1 dx_2,\]
which yields,
\[V(X)=\int (x_1-E(X_1))^2f_1(x_1)\,dx_1+\int (x_2-E(X_2))^2f_2(x_2)\,dx_2=V(X_1)+V(X_2)=\frac{1}{12}.\]
Hence the lemma.
\end{proof}
\begin{note} \label{note1} We have $E(X_1)=\frac 1 2$ and $E(X_2)=\frac{\sqrt 3}{6}$, and so by the standard rule of probability theory, for any two real numbers $a$ and $b$, we deduce
$E(X_1-a)^2=E(X_1-\frac 1 2)^2+(a-\frac 1 2)^2=V(X_1)+(a-\frac 1 2)^2$, and similarly
$E(X_2-b)^2=V(X_2)+(b-\frac{\sqrt 3}{6})^2$. Thus, for any $(a, b) \in \D R^2$, we have $E\|X-(a, b)\|^2=\iint_{\D R^2} [(x_1-a)^2+(x_2-b)^2]f(x_1, x_2)dx_1dx_2=\int_{\D R} (x_1-a)^2 f_1(x_1) dx_1+\int_{\D R} (x_2-b)^2f_2(x_2) dx_2=E(X_1-a)^2 +E(X_2-b)^2=V(X_1)+V(X_2)+(a-\frac 1 2)^2+(b-\frac{\sqrt 3}{6})^2=\frac 1 {12}+\|(a, b)-(\frac 12, \frac{\sqrt 3}{6})\|^2$.

\end{note}
\begin{note}
From Note~\ref{note1} it is clear that the optimal set of one-mean consists of the expected vector $(\frac 12, \frac {\sqrt 3}{6})$ of the random variable $X$, which is the centroid of the triangle $\tri$ and the corresponding quantization error is  $\frac 1 {12}$, which is the expected squared distance of the random variable $X$.
\end{note}

\section{Optimal sets of 2-means}

\begin{figure}
\begin{tikzpicture}[line cap=round,line join=round,>=triangle 45,x=3.0cm,y=3.0cm]
\clip(-0.1,-0.2) rectangle (1.1,1.);
\draw (0.5,0.866025)-- (0.,0.);
\draw (0.5,0.866025)-- (1.,0.);
\draw [color=ffqqqq] (0.30901704553129616,0.5352329737124816)-- (0.618034,0.);
\draw [dotted] (0.6197184621526524,0.6586666376284985)-- (0.618034,0.);
\draw [dotted] (0.5,0.866025)-- (0.49121462970320096,0.);
\draw [dotted] (0.30901704553129616,0.5352329737124816)-- (0.31350534220513926,0.);
\draw (0.,0.)-- (1.,0.);
\draw (0.35559438398099597,0.5280204116031809) node[anchor=north west] {$\ell$};
\begin{scriptsize}
\draw [fill=uuuuuu] (0.,0.) circle (1.5pt);
\draw[color=uuuuuu] (-0.013853871607079718,-0.06590273345613046) node {$O$};
\draw [fill=xdxdff] (1.,0.) circle (1.5pt);
\draw[color=xdxdff] (1.0524018533812907,0.06036439187143966) node {$A$};
\draw [fill=qqqqff] (0.5,0.866025) circle (1.5pt);
\draw[color=qqqqff] (0.5333036714790577,0.9302045885724783) node {$B$};
\draw [fill=xdxdff] (0.618034,0.) circle (1.5pt);
\draw[color=xdxdff] (0.6128051948334537,-0.07057929365344788) node {$C$};
\draw [fill=xdxdff] (0.30901704553129616,0.5352329737124816) circle (1.5pt);
\draw[color=xdxdff] (0.2246506984561084,0.5794625737736724) node {$D$};
\draw [fill=qqqqff] (0.309017,0.178411) circle (1.5pt);
\draw[color=qqqqff] (0.23868037904806064,0.17260183660705755) node {$P$};
\draw [fill=qqqqff] (0.618034,0.356822) circle (1.5pt);
\draw[color=qqqqff] (0.682953597793215,0.3877236056836585) node {$Q$};
\draw [fill=xdxdff] (0.31350534220513926,0.) circle (1.5pt);
\draw [fill=xdxdff] (0.6197184621526524,0.6586666376284985) circle (1.5pt);
\draw [fill=xdxdff] (0.49121462970320096,0.) circle (1.5pt);
\end{scriptsize}
\end{tikzpicture}
\caption{Optimal configuration of two points $P$ and $Q$.}
\end{figure}
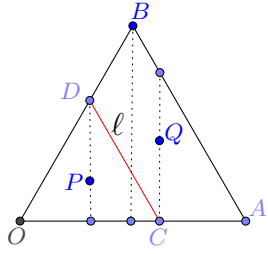

In this section we obtain all the optimal sets of two-means and the corresponding quantization error.
Let $\tri$ be the equilateral triangle with vertices $O (0, 0), \, A (1, 0)$, and $B (\frac 1 2, \frac{\sqrt 3}{2})$. Let us divide the triangle $\tri$ by a straight line $\ell$ into two regions. Let us first assume that the vertex $O$ is in one side of $\ell$ and the vertices $A$ and $B$ are in the other side of $\ell$. It might be that one of $A$ and $B$ lies on the line $\ell$.  Thus, the triangle $\tri$ is divided into two regions: the triangle $OCD$ and the quadrilateral $CABD$, where $C$ and $D$ are the points of intersections of the line with the sides $OA$ and $OB$ respectively. If either $A$ or $B$ is on the line $\ell$, then $CABD$ will also be a triangle. Let $P$ and $Q$ be the centroids of the regions $OCD$ and $CABD$ respectively. Let the position vectors of $A, \, B, \, P, \, Q, \, C, \, D$ be denoted respectively by $\tilde a$, $\tilde b$, $\tilde p$, $\tilde q$, $\tilde c$, $\tilde d$. Then, there exist scalars $\ga$ and $\gb$ such that $ \tilde c=\ga \tilde a, \,  \tilde d=\gb \tilde b, \, \tilde p=\frac  1 3 (\tilde c+\tilde d)=\frac 1 3(\ga \tilde a+\gb \tilde b)$, and  the area of the triangle $ OCD=\frac{\sqrt{3}}{4} \ga \gb.$
Since the probability measure is uniformly distributed over $\tri$, taking moments about the origin, we have
\[\tilde q=\frac{\frac 1 3 (\tilde a +\tilde b)\frac{\sqrt{3}}{4}  -\frac 1 3(\ga \tilde a+\gb \tilde b) \frac{\sqrt{3}}{4}\ga \gb}{\frac{\sqrt{3}}{4}-\frac{\sqrt{3}}{4} \ga \gb}=\frac{\tilde a +\tilde b-\ga \gb(\ga \tilde a+\gb\tilde b)}{3(1-\ga \gb)}.\]
If $P$ and $Q$ form an optimal set of two-means, then $CD$ will be the boundary of their corresponding Voronoi regions, and so we have
$|\overrightarrow{CP}|=|\overrightarrow{CQ}| \te { and } |\overrightarrow{DP}|=|\overrightarrow{DQ}|, \te{ i.e., } (\overrightarrow{CP})^2=(\overrightarrow{CQ})^2 \te { and } (\overrightarrow{DP})^2=(\overrightarrow{DQ})^2.$
Using the dot product of vectors, we have  $\tilde a^2=\tilde b^2=1$ and $\tilde a \cdot \tilde b=1\cdot 1\cdot\cos \frac{\pi}{3}=\frac 1 2$.
 Then,  $(\overrightarrow{CP})^2=(\overrightarrow{CQ})^2$ implies
\[\Big(\frac 1 3(\ga \tilde a+\gb \tilde b)-\ga \tilde a\Big)^2=\Big(\frac{\tilde a +\tilde b-\ga \gb(\ga \tilde a+\gb\tilde b)}{3(1-\ga \gb)}-\ga \tilde a\Big)^2\]
which after simplification yields
\begin{equation} \label{eq123} 4 \alpha^3 \beta +\alpha^2 \beta^2-6 \alpha^2 \beta -5 \alpha^2-2 \alpha  \beta^3+3 \alpha  \beta^2-2 \alpha  \beta +9 \alpha +\beta^2-3=0.\end{equation}
Due to symmetry, $(\overrightarrow{DP})^2=(\overrightarrow{DQ})^2$ yields,
\begin{equation} \label{eq124} 4 \alpha \beta^3 +\alpha^2 \beta^2-6 \alpha \beta^2 -5 \beta^2-2 \alpha^3  \beta+3 \alpha^2  \beta-2 \alpha  \beta +9 \beta +\alpha^2-3=0.\end{equation}
Solving \eqref{eq123} and \eqref{eq124}, we get the five sets of solutions for $\ga$ and $\gb$:
$\set{\alpha =\frac{1}{2}, \, \beta=1}, \, \set{\alpha = 1, \, \beta =\frac{1}{2}},\, \set{\alpha = 1,\, \beta = 1}, \, \{\alpha = \frac{1}{2} (-1-\sqrt{5}),\, \beta = \frac{1}{2} (-1-\sqrt{5})\}, \te{ and } \{\alpha = \frac{1}{2} (\sqrt{5}-1), \, \beta = \frac{1}{2} (\sqrt{5}-1)\},$
among which the admissible solutions are
$\set {\alpha =\frac{1}{2}, \, \beta=1}, \, \{\alpha = 1, \, \beta =\frac{1}{2}\}, \, \{\alpha = \frac{1}{2} (\sqrt{5}-1), \, \beta = \frac{1}{2} (\sqrt{5}-1)\}.$
If $\set{\alpha =\frac{1}{2}, \, \beta=1}$, then the line $\ell$ passes through the vertex $B$, and if $\set{\alpha = 1, \, \beta =\frac{1}{2}}$, then the line $\ell$ passes through the vertex $A$. Let us first take $\{\alpha =\frac{1}{2}, \, \beta=1\}$. Then, $\tilde p=( \frac{1}{3}, \frac{1}{2 \sqrt{3}})$ and $\tilde q=( \frac{2}{3}, \frac{1}{2 \sqrt{3}})$, and the corresponding quantization error
\begin{align*}=&\int _0^{\frac{1}{2}}\int _0^{\sqrt{3} x_1}\frac{4 ((x_1-\frac{1}{3}){}^2+(x_2-\frac{1}{2 \sqrt{3}}){}^2)}{\sqrt{3}}dx_2dx_1\\
&+\int _{\frac{1}{2}}^1\int _0^{-\sqrt{3} (x_1-1)}\frac{4 ((x_1-\frac{2}{3}){}^2+(x_2-\frac{1}{2 \sqrt{3}}){}^2)}{\sqrt{3}}dx_2dx_1=\frac{1}{18}=0.0555556.
\end{align*}
Similarly, it can be shown that if $\{\alpha = 1, \, \beta =\frac{1}{2}\}$, then the quantization error is $0.0555556$. Now take $\ga=\gb=\frac{1}{2} (\sqrt{5}-1)$. Then,
$\tilde p=( 0.309017, 0.178411)  \te{ and }  \tilde q=( 0.618034, 0.356822) $, and the corresponding quantization error
\begin{align*}
&=\int _0^{\frac{1}{4} (\sqrt{5}-1}\int _0^{\sqrt{3} x_1}\frac{4 ((x_1-0.309017){}^2+(x_2-0.178411){}^2)}{\sqrt{3}}dx_2dx_1\\
&+\int _{\frac{1}{4} (\sqrt{5}-1)}^{\frac{1}{2} (\sqrt{5}-1)}\int _0^{\frac{1}{2} (\sqrt{15}-\sqrt{3})-\sqrt{3} x_1}\frac{4 ((x_1-0.309017){}^2+(x_2-0.178411){}^2)}{\sqrt{3}}dx_2dx_1\\
&+\int _{\frac{1}{4} (\sqrt{5}-1)}^{\frac{1}{2}}\int _{\frac{1}{2} (\sqrt{15}-\sqrt{3})-\sqrt{3} x_1}^{\sqrt{3} x_1}\frac{4 ((x_1-0.618034){}^2+(x_2-0.356822){}^2)}{\sqrt{3}}dx_2dx_1\\
&+\int _{\frac{1}{2}}^{\frac{1}{2} (\sqrt{5}-1)}\int _{\frac{1}{2} (\sqrt{15}-\sqrt{3})-\sqrt{3} x_1}^{-\sqrt{3} (x_1-1)}\frac{4 ((x_1-0.618034){}^2+(x_2-0.356822){}^2)}{\sqrt{3}}dx_2dx_1\\
&+\int _{\frac{1}{2} (\sqrt{5}-1)}^1\int _0^{-\sqrt{3} (x_1-1)}\frac{4 ((x_1-0.618034){}^2+(x_2-0.356822){}^2)}{\sqrt{3}}dx_2dx_1\\
&=0.0532767.
\end{align*}
Since $0.0532767<0.0555556$, an optimal set of two-means is obtained for $\ga=\gb= \frac{1}{2} (\sqrt{5}-1)$, i.e., the set $\set{(0.309017, 0.178411),\, (0.618034, 0.356822)}$ forms an optimal set of two-means, and the two means lie on the median passing through the vertex $O$ (see Figure 1). Notice that $g^{-1}=\frac{1}{2} (\sqrt{5}-1)$, where $g:=\frac{\sqrt 5+1}{2}$ is the golden ratio. Since $\ga=\gb=g^{-1}$, we can say that the line $\ell$ is parallel to the side $AB$, and cuts the triangle $\tri$ into an equilateral triangle and an isosceles trapezoid. Due to symmetry, the line $\ell$ can also be parallel to either $OA$ or $OB$, i.e., the two means can also lie either on the median passing through the vertex $B$, or on the median passing through the vertex $A$. Moreover, it can be seen that
\[\frac{\te{Area of the isosceles trapezoid } CABD }{\te{Area of the equilateral triangle } OCD}=\frac {\frac{1}{8} \sqrt{3} (\sqrt{5}-1)}{\frac{1}{8} \sqrt{3}(3-\sqrt 5)}=\frac{\sqrt 5-1}{3-\sqrt{5}}=\frac{g^2}{g}=g.\]
Therefore, we can deduce the following theorem.

\begin{theorem} \label{th2}
Let $X$ be a random variable with uniform distribution on the equilateral triangle $\tri$ with vertices $(0, 0), \, (1, 0)$, and $(\frac 1 2, \frac{\sqrt 3}{2})$. Then, there are three optimal sets of two-means with quantization error $0.0532767$. If the triangle $\tri$ is partitioned into an isosceles trapezoid and an equilateral triangle in the golden ratio, then the centroids of the isosceles trapezoid and the equilateral triangle form an optimal set of two-means.
\end{theorem}

\section{Optimal set of 3-means}

\begin{theorem}
For uniform distribution on the equilateral triangle with vertices $(0, 0)$, $(1, 0)$ and $(\frac{1}{2}, \frac{\sqrt{3}}{2})$, the set $\set{(\frac{7}{24}, \frac{7}{24 \sqrt{3}}), \, (\frac{17}{24}, \frac{7}{24 \sqrt{3}}), (\frac{1}{2}, \frac{11}{12 \sqrt{3}})}$ is the only optimal set of three-means. The three means in this case form an equilateral triangle having the sides parallel to the sides of the original triangle.
\end{theorem}

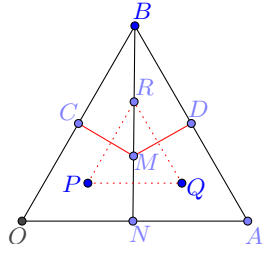
\begin{figure}
\begin{tikzpicture}[line cap=round,line join=round,>=triangle 45,x=3.0cm,y=3.0cm]
\clip(-0.1,-0.1) rectangle (1.1,1.);
\draw (0.5,0.866025)-- (0.,0.);
\draw (0.5,0.866025)-- (1.,0.);
\draw (0.5,0.866025)-- (0.49121462970320096,0.);
\draw (0.,0.)-- (1.,0.);
\draw [color=ffqqqq] (0.2500002165064014,0.4330128749999126)-- (0.49414368914411855,0.2887344092080969);
\draw [color=ffqqqq] (0.7499997834935986,0.43301287499991253)-- (0.49414368914411855,0.2887344092080969);
\draw [dotted,color=ffqqqq] (0.291667,0.168394)-- (0.4965838237391188,0.5292726553199399);
\draw [dotted,color=ffqqqq] (0.708333,0.168394)-- (0.4965838237391188,0.5292726553199399);
\draw [dotted,color=ffqqqq] (0.291667,0.168394)-- (0.708333,0.168394);
\begin{scriptsize}
\draw [fill=uuuuuu] (0.,0.) circle (1.5pt);
\draw[color=uuuuuu] (-0.018530431804397132,-0.06522617325881305) node {$O$};
\draw [fill=xdxdff] (1.,0.) circle (1.5pt);
\draw[color=xdxdff] (1.0280190523947036,-0.06522617325881305) node {$A$};
\draw [fill=qqqqff] (0.5,0.866025) circle (1.5pt);
\draw[color=qqqqff] (0.5333036714790577,0.9302045885724783) node {$B$};
\draw [fill=xdxdff] (0.2500002165064014,0.4330128749999126) circle (1.5pt);
\draw[color=xdxdff] (0.20723821687756909,0.4859313698273242) node {$C$};
\draw [fill=qqqqff] (0.291667,0.168394) circle (1.5pt);
\draw[color=qqqqff] (0.21997413825879097,0.16324871621242273) node {$P$};
\draw [fill=qqqqff] (0.708333,0.168394) circle (1.5pt);
\draw[color=qqqqff] (0.7718082415422458,0.14224871621242273) node {$Q$};
\draw [fill=xdxdff] (0.7499997834935986,0.43301287499991253) circle (1.5pt);
\draw[color=xdxdff] (0.7811613619368806,0.4999610504192764) node {$D$};
\draw [fill=xdxdff] (0.49121462970320096,0.) circle (1.5pt);
\draw[color=xdxdff] (0.5239505510844228,-0.05707305286417823) node {$N$};
\draw [fill=xdxdff] (0.4965838237391188,0.5292726553199399) circle (1.5pt);
\draw[color=xdxdff] (0.5426567918736925,0.5934922543656246) node {$R$};
\draw [fill=xdxdff] (0.49414368914411855,0.2887344092080969) circle (1.5pt);
\draw[color=xdxdff] (0.5520099122683273,0.26145648035608837) node {$M$};
\end{scriptsize}
\end{tikzpicture}
\caption{Optimal configuration of three points $P$, $Q$ and $R$.}
\end{figure}

\begin{proof}
Due to symmetry of the triangle with the uniform distribution, we can assume that one element in the optimal set of three-means lies on a median of the triangle, and the other two are equidistant from the median. As shown in Figure 2, let the median passing through the vertex $B$ cuts the side $OA$ at the point $N$, and let one element in the optimal set of three-means lie on this median. Let the boundaries of the Voronoi regions cut the sides $OB$ and $AB$ at the points $C$ and $D$ respectively. Let the three boundaries of the Voronoi regions meet at the point $M$ which lies on the median $BN$. Let the position vectors of the points $A,\, B, \, C, \, D, \, M, \, N$ be respectively $\tilde a, \, \tilde b, \, \tilde c, \, \tilde d, \, \tilde m, \, \tilde n$. Let $\ga$ and $\gb$ be two scalars such that the length of $BC$ equals $\ga$ and the length of $BM$ equals $\frac {\sqrt{3}}{2} \gb$. Due to symmetry, the length of $BD$ is also $\ga$. Then,
$\tilde  c=(1-\ga) \tilde b, \,  \tilde d  =\ga \tilde a +(1-\ga)\tilde b, \te{ and } \tilde m=\gb \tilde  n +(1-\gb) \tilde b.$
$\te{ Area of the triangle } BCM =\te{Area of the triangle } BDM=\frac{\sqrt{3}}{8} \ga \gb.$
Let the centroids of the quadrilaterals $ONMC$, $NADM$,  and $BCMD$ be $P$, $Q$, and $R$ with position vectors $\tilde p$, $\tilde q$, and $\tilde r$ respectively. Since the probability measure is uniformly distributed over $\tri$, taking moments about the origin, we have
\begin{align*}
\tilde p & =\frac{\frac 1 3 (\tilde b+\tilde n)\frac{\sqrt{3}}{8}- \frac 1 3 (\tilde b+\tilde c+\tilde m)\frac{\sqrt{3}}{8} \ga \gb}{\frac{\sqrt{3}}{8}-\frac{\sqrt{3}}{8} \ga \gb}=\frac{ \tilde b +\tilde n - (\tilde b +\tilde c+\tilde m)\ga \gb}{3(1-\ga \gb)},\\
\tilde q&=\frac{\frac 1 3 (\tilde a+\tilde b+\tilde n)\frac{\sqrt{3}}{8}- \frac 1 3 (\tilde b+\tilde d+\tilde m)\frac{\sqrt{3}}{8} \ga \gb}{\frac{\sqrt{3}}{8}-\frac{\sqrt{3}}{8} \ga \gb}=\frac{ \tilde a+\tilde b +\tilde n - (\tilde b +\tilde d+\tilde m)\ga \gb}{3(1-\ga \gb)},\\
\tilde r& =\frac{\frac 1 3 (\tilde b+\tilde c+\tilde m)\frac{\sqrt{3}}{8} \ga \gb+ \frac 1 3 (\tilde b+\tilde d+\tilde m)\frac{\sqrt{3}}{8} \ga \gb}{\frac{\sqrt{3}}{4} \ga \gb}=\frac{ \tilde c +\tilde d +2 (\tilde b +\tilde m)}{6}.
\end{align*}
If $P$, $Q$ and $R$ be the optimal points, we must have
$(\overrightarrow{RC})^2=(\overrightarrow{PC})^2, \, (\overrightarrow{RM})^2=(\overrightarrow{PM})^2, \, (\overrightarrow{RM})^2=(\overrightarrow{QM})^2 \te{ and } (\overrightarrow{RD})^2=(\overrightarrow{QD})^2. $
Using the dot product of vectors, we have
$\tilde a^2=1,\, \tilde b^2= 1, \, \tilde n^2=\frac{1}{4},\, \tilde a\cdot\tilde n=\frac{1}{2}, \, \tilde a\cdot\tilde b=\frac{1}{2},\, \tilde b \cdot \tilde n= \frac{1}{4}.$
Then, $(\overrightarrow{RC})^2=(\overrightarrow{PC})^2$ implies,
\[\Big((1-\alpha ) \tilde b-\frac{ \tilde c +\tilde d +2 (\tilde b +\tilde m)}{6}\Big)^2=\Big((1-\alpha ) \tilde b-\frac{ \tilde b +\tilde n - (\tilde b +\tilde c+\tilde m)\ga \gb}{3(1-\ga \gb)}\Big)^2,\]
which after simplification yields
\begin{equation} \label{eq1001}
5 \alpha ^4 \beta ^2+6 \alpha ^3 \beta +\alpha ^2 (6 \beta ^2-28 \beta -15)-6 \alpha  (\beta ^3-2 \beta ^2+2 \beta -7)+3 \beta ^2-13=0.\end{equation}
$ (\overrightarrow{RM})^2=(\overrightarrow{PM})^2$ implies
\[(\gb \tilde n +(1-\gb) \tilde b -\frac{ \tilde c +\tilde d +2 (\tilde b +\tilde m)}{6})^2=(\gb \tilde n +(1-\gb) \tilde b-\frac{ \tilde b +\tilde n - (\tilde b +\tilde c+\tilde m)\ga \gb}{3(1-\ga \gb)})^2\] which after simplification yields
\begin{equation} \label{eq1002}
\alpha ^4 \left(-\beta ^2\right)-6 \alpha ^3 \beta +\alpha ^2 \left(6 \beta ^2+14 \beta +3\right)+12 \alpha  \beta  \left(\beta ^2-2 \beta -1\right)-15 \beta ^2+36 \beta -13=0.\end{equation}
Solving the equations \eqref{eq1001} and \eqref{eq1002}, we have $\ga=\frac 12 $ and $\gb=\frac 23$. Then, we have
$\tilde p=(\frac{7}{24},\frac{7}{24 \sqrt{3}}), \, \tilde q=(\frac{17}{24},\frac{7}{24 \sqrt{3}}), \te{ and } \tilde r=(\frac{1}{2},\frac{11}{12 \sqrt{3}})$. Moreover, $\tilde c=(\frac{1}{4},\frac{\sqrt{3}}{4})$ and $\tilde d=(\frac{3}{4},\frac{\sqrt{3}}{4})$. Here the equation of the line $OB$ is $x_2=\sqrt{3} x_1$, and the equation of the line $CM$ is $x_2=-\frac{x_1-1}{\sqrt{3}}$. Thus, if $V_3(P)$ is the quantization error due to the point $P$ in its Voronoi region, then we have
\begin{align*}
&V_3(P)=\int _0^{\frac{1}{4}}\int _0^{\sqrt{3} x_1}\frac{4 ((x_1-\frac{7}{24}){}^2+(x_2-\frac{7}{24 \sqrt{3}}){}^2)}{\sqrt{3}}dx_2dx_1\\
&+\int _{\frac{1}{4}}^{\frac{1}{2}}\int _0^{-\frac{x_1-1}{\sqrt{3}}}\frac{4 ((x_1-\frac{7}{24}){}^2+(x_2-\frac{7}{24 \sqrt{3}}){}^2)}{\sqrt{3}}dx_2dx_1=\frac{11}{1296}.
\end{align*}
 Due to the uniform distribution and the symmetry of the points, we have $V_3(P)=V_3(Q)=V_3(R)=\frac{11}{1296}$. Thus, the set $\set{(\frac{7}{24}, \frac{7}{24 \sqrt{3}}), \, (\frac{17}{24}, \frac{7}{24 \sqrt{3}}), (\frac{1}{2}, \frac{11}{12 \sqrt{3}})}$ forms an optimal set of three-means with quantization error $V_3=3 \times \frac{11}{1296}=\frac{11}{432}$. Notice that the points $(\frac{7}{24}, \frac{7}{24 \sqrt{3}})$ and $(\frac{17}{24}, \frac{7}{24 \sqrt{3}})$ lie on the medians passing through the vertices $O$ and $A$ respectively, and the three points in this case form an equilateral triangle having the sides parallel to the sides of the original triangle.  Thus, due to symmetry, we can say that the set $\set{(\frac{7}{24}, \frac{7}{24 \sqrt{3}}), \, (\frac{17}{24}, \frac{7}{24 \sqrt{3}}), (\frac{1}{2}, \frac{11}{12 \sqrt{3}})}$ is the only optimal set of three-means. Hence, the proof of the theorem is complete.
\end{proof}

\section{Optimal sets of 4-means}

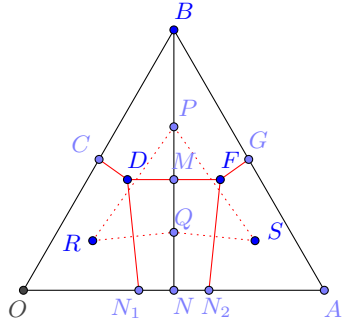
\begin{figure}
\begin{tikzpicture}[line cap=round,line join=round,>=triangle 45,x=4.0cm,y=4.0cm]
\clip(-0.1,-0.15) rectangle (1.1,1.);
\draw (0.5,0.866025)-- (0.,0.);
\draw (0.5,0.866025)-- (1.,0.);
\draw (0.5,0.866025)-- (0.5,0.);
\draw (0.,0.)-- (1.,0.);
\draw [color=ffqqqq] (0.2513528339815584,0.43535567609775816)-- (0.3456800438177196,0.36816779157149887);
\draw [color=ffqqqq] (0.3456800438177196,0.36816779157149887)-- (0.6543199561822803,0.36816779157149887);
\draw [color=ffqqqq] (0.7486471660184416,0.43535567609775816)-- (0.6543199561822803,0.36816779157149887);
\draw [color=ffqqqq] (0.3456800438177196,0.36816779157149887)-- (0.3834684123722554,0.);
\draw [color=ffqqqq] (0.6543199561822803,0.36816779157149887)-- (0.6165315876277446,0.);
\draw [dotted,color=ffqqqq] (0.5,0.5437239171503262)-- (0.23023301493672835,0.16495622450758732);
\draw [dotted,color=ffqqqq] (0.23023301493672835,0.16495622450758732)-- (0.5,0.19271412482179784);
\draw [dotted,color=ffqqqq] (0.7697669850632717,0.16495622450758732)-- (0.5,0.19271412482179784);
\draw [dotted,color=ffqqqq] (0.5,0.5437239171503262)-- (0.7697669850632717,0.16495622450758732);
\begin{scriptsize}
\draw [fill=uuuuuu] (0.,0.) circle (1.5pt);
\draw[color=uuuuuu] (-0.018530431804397132,-0.06122617325881305) node {$O$};
\draw [fill=xdxdff] (1.,0.) circle (1.5pt);
\draw[color=xdxdff] (1.0290190523947036,-0.06122617325881305) node {$A$};
\draw [fill=qqqqff] (0.5,0.866025) circle (1.5pt);
\draw[color=qqqqff] (0.5333036714790577,0.9302045885724783) node {$B$};
\draw [fill=xdxdff] (0.2513528339815584,0.43535567609775816) circle (1.5pt);
\draw[color=xdxdff] (0.1919147770748865,0.4859313698273242) node {$C$};
\draw [fill=qqqqff] (0.23023301493672835,0.16495622450758732) circle (1.5pt);
\draw[color=qqqqff] (0.1591788556936646,0.15857215601510533) node {$R$};
\draw [fill=qqqqff] (0.7697669850632717,0.16495622450758732) circle (1.5pt);
\draw[color=qqqqff] (0.8372800843046896,0.19598463759364462) node {$S$};
\draw [fill=xdxdff] (0.7486471660184416,0.43535567609775816) circle (1.5pt);
\draw[color=xdxdff] (0.7811613619368806,0.4999610504192764) node {$G$};
\draw [fill=xdxdff] (0.5,0.) circle (1.5pt);
\draw[color=xdxdff] (0.5333036714790577,-0.05187305286417823) node {$N$};
\draw [fill=xdxdff] (0.5,0.5437239171503262) circle (1.5pt);
\draw[color=xdxdff] (0.5473333520710099,0.6075219349575769) node {$P$};
\draw [fill=xdxdff] (0.5,0.36821902098606213) circle (1.5pt);
\draw[color=xdxdff] (0.5333036714790577,0.4298126474595152) node {$M$};
\draw [fill=xdxdff] (0.5,0.19271412482179784) circle (1.5pt);
\draw[color=xdxdff] (0.5333036714790577,0.23807367936950133) node {$Q$};
\draw [fill=qqqqff] (0.3456800438177196,0.36816779157149887) circle (1.5pt);
\draw[color=qqqqff] (0.37897718496758304,0.43448920765683263) node {$D$};
\draw [fill=qqqqff] (0.6543199561822803,0.36816779157149887) circle (1.5pt);
\draw[color=qqqqff] (0.6876301579905323,0.43448920765683263) node {$F$};
\draw [fill=xdxdff] (0.3834684123722554,0.) circle (1.5pt);
\draw[color=xdxdff] (0.3415647033890437,-0.06122617325881305) node {$N_1$};
\draw [fill=xdxdff] (0.6165315876277446,0.) circle (1.5pt);
\draw[color=xdxdff] (0.6408645560173583,-0.05654961306149564) node {$N_2$};
\end{scriptsize}
\end{tikzpicture}
\caption{Optimal configuration of four points $P$, $Q$, $R$ and $S$.}
\end{figure}

In this section we calculate the optimal sets of four-means. Let $OAB$ be the equilateral triangle with vertices $(0, 0)$, $(1, 0)$ and $(\frac 1 2, \frac{\sqrt 3} {2})$. As shown in Figure 3, let $BN$ be the median of the triangle passing through the vertex $B$ \
which cuts $OA$ at the point $N$. Let $\set{P, \, Q, \, R, \, S}$
be an optimal set of four-means, where $P, \, Q$ are on the median $BN$; and $R, \, S$
are in the opposite sides of the median. Notice that, our assumption is also verified by a numerical search algorithm as mentioned in the next section.
Let $CD$ be the boundary of the Voronoi
regions of the points $P$ and $R$, $DF$ be the boundary of the Voronoi
regions of the points $P$ and $Q$ which cuts the median $BN$ at the point $M$,
$FG$ be the boundary of the Voronoi regions of the points $P$ and $S$.
Let $DN_1$ and $FN_2$ be the boundaries of the
Voronoi regions of the points $R$, $Q$ and  $Q$, $S$ respectively. Let $\ga$, $\gb$, $\gg$, $\gd$  be four constants such
that  $BC = BG = \ga$, $ON_1=AN_2=\gd$, $BM = \frac{\sqrt 3}{2}\gb$; $x_1$-coordinate of $D$ be $\gg$,
and so due to symmetry $x_1$-coordinate of $F$ is $1 -\gg$. Then we have,

\[ \begin{array}{llll}
  \tilde a=(1, 0), & \tilde b=( \frac{1}{2},\frac{\sqrt{3}}{2}), & \tilde n=( \frac{1}{2},0 ),  \\
\tilde c=(1-\alpha ) \tilde b, & \tilde d=( \gamma ,\frac{1}{2} \sqrt{3} (1-\beta )), & \tilde g=\tilde a \alpha +(1-\alpha ) \tilde b,\\
\tilde m=\tilde b (1-\beta )+\beta  \tilde n, & \tilde n_1=( \delta ,0), & \tilde n_2=( 1-\delta ,0),\\
 \tilde f=( 1-\gamma ,\frac{1}{2} \sqrt{3} (1-\beta )).
\end{array}
\]
The equation of the line $CD$ is $x_2=\frac{1}{2} \sqrt{3} (1-\beta )+\frac{\sqrt{3} (\alpha -\beta ) (x_1-\gamma )}{\alpha +2 \gamma -1}$. The equation of the line $BD$ is $x_2=\frac{\sqrt{3} \beta  (x_1-\frac{1}{2})}{1-2 \gamma }+\frac{\sqrt{3}}{2}$. If Ar$_1$ is the area of the triangle $BCD$, then
\begin{align*} \text{Ar}_1&=\int _{\frac{1-\alpha }{2}}^{\gamma }\int _{\frac{1}{2} \sqrt{3} (1-\beta )+\frac{\sqrt{3} (\alpha -\beta ) (x_1-\gamma)}{\alpha +2 \gamma -1}}^{\sqrt{3} x_1}1\,dx_2dx_1+\int _{\gamma }^{\frac{1}{2}}\int _{\frac{\sqrt{3} \beta (x_1-\frac{1}{2})}{1-2 \gamma }+\frac{\sqrt{3}}{2}}^{\sqrt{3} x_1}1 \,dx_2dx_1\\
&=\frac{1}{8} \sqrt{3} (\alpha +2 \gamma -1) (\beta +2 \gamma -1)-\frac{1}{4} \sqrt{3} \beta  \gamma -\frac{\sqrt{3} \beta  \gamma }{2 (1-2 \gamma )}+\frac{\sqrt{3} \beta }{4 (1-2 \gamma )}-\frac{\sqrt{3} \beta }{8}-\frac{1}{2} \sqrt{3} \gamma ^2\\
& \qquad \qquad +\frac{\sqrt{3} \gamma }{2}-\frac{\sqrt{3}}{8}.\end{align*}
If Ar$_2$ is the area of the triangle $BDF$, then
$\te{Ar}_2=\frac{\sqrt{3} \beta  (1-2 \gamma )}{2\cdot 2}=\frac{1}{4} \sqrt{3} \beta  (1-2 \gamma ).$
If Ar$_3$ is the area of the triangle $BFG$, then $\te{Ar}_3=\te{Ar}_1.$
If  Ar$_4$ is the area of the triangle $OCD$, then
\begin{align*}
\te{Ar}_4&=\int _0^{\frac{1-\alpha }{2}}\int _{\frac{\sqrt{3} (1-\beta ) x_1}{2 \gamma }}^{\sqrt{3} x_1}1\,dx_2dx_1+\int _{\frac{1-\alpha }{2}}^{\gamma }\int _{\frac{\sqrt{3} (1-\beta ) x_1}{2 \gamma }}^{\frac{1}{2} \sqrt{3} (1-\beta )+\frac{\sqrt{3} (\alpha -\beta ) (x_1-\gamma )}{\alpha +2 \gamma -1}}1\,dx_2dx_1\\
&=\frac{\sqrt{3} (\alpha -1)^2 (\beta +2 \gamma -1)}{16 \gamma }-\frac{\sqrt{3} (\alpha -1) (\alpha +2 \gamma -1) (\beta +2 \gamma -1)}{16 \gamma }.
\end{align*}
If Ar$_5$ is the area of the triangle $ODN_1$,
then $\text{Ar}_5=\frac{\sqrt{3} (1-\beta ) \delta }{2\cdot 2}.$
If Ar$_6$ is the area of the triangle $DN_1N_2$, then
$\text{Ar}_6= \frac{\sqrt{3} (1-\beta ) (1-2 \delta )}{2\cdot 2}=\frac{1}{4} \sqrt{3} (1-\beta ) (1-2 \delta ).$
If Ar$_7$ is the area of the triangle $DFN_2$, then
$\text{Ar}_7=\frac{\sqrt{3} (1-\beta ) (1-2 \gamma )}{2\cdot 2}=\frac{1}{4} \sqrt{3} (1-\beta ) (1-2 \gamma ).$
Notice  that due to symmetry, if Ar$_8$ is the area of the triangle $FN_2A$ and Ar$_9$ the area of the triangle $FAG$, then
$\te{Ar}_8=\te{Ar}_5  \te{ and } \te{Ar}_9=\te{Ar}_4.$
As $P$, $Q$, $R$, $S$ are assumed to form an optimal set of four-means, they are also the centroids of their corresponding Voronoi regions associated with the density function $f(x_1, x_2)$ which is constant due to the uniform distribution over the triangle.
Thus, $P, \, Q, \, R, \, S$ are respectively the centroids of the pentagon $BCDFG$, quadrilaterals $DN_1N_2F$, $OCDN_1$, and $AN_2FG$. Hence,  we have
\begin{align*}
\tilde p&=\frac{\frac{1}{3} \text{Ar}_1 (\tilde b+\tilde c+\tilde d)+\frac{1}{3} \text{Ar}_2 (\tilde b+\tilde d+\tilde f)+\frac{1}{3} \text{Ar}_3 (\tilde b+\tilde f+\tilde g)}{\text{Ar}_1+\text{Ar}_2+\text{Ar}_3},\\
\tilde q &=\frac{\frac{1}{3} \text{Ar}_7 (\tilde d+\tilde f+\tilde n_2)+\frac{1}{3} \text{Ar}_6 (\tilde d+\tilde n_1+\tilde n_2)}{\text{Ar}_6+\text{Int}_7},\\
\tilde r&=\frac{\frac{1}{3} \text{Ar}_4 (\tilde c+\tilde d)+\frac{1}{3} \text{Ar}_5 (\tilde d+\tilde  n_1)}{\text{Ar}_4+\text{Ar}_5},\\
\tilde s&=\frac{\frac{1}{3} \text{Ar}_9 (\tilde a+\tilde f+\tilde g)+\frac{1}{3} \text{Ar}_8 (\tilde a+\tilde f+\tilde n_2)}{\text{Ar}_8+\text{Ar}_9}.
\end{align*}
Write $
\text{Q1}:=\rho (\tilde p, \tilde c)-\rho(\tilde c, \tilde r), \, \text{Q2}: =\rho(\tilde p, \tilde d)-\rho (\tilde d, \tilde r), \, \text{Q3}:=\rho (\tilde q, \tilde d)-\rho(\tilde d, \tilde r)$, and $\text{Q4}:=\rho(\tilde q, \tilde n_1)-\rho(\tilde n_1, \tilde r).
$
Since the line passing through the boundary of the Voronoi regions of any two points in an optimal set of $n$-means, $n\geq 2$, is the perpendicular bisector of the line segment joining the two points, we must have $Q1=0$, $Q2=0$, $Q3=0$ and $Q4=0$.   Using Mathematica, we solve these four equations for the parameters $\ga, \, \gb, \, \gg$ and $\gd$ up to 20 decimal places and obtain
\begin{align*}
&\alpha= 0.49729450782679201845, \, \beta=0.57487645285849021867,\\
&\gamma= 0.34568004381771961464, \, \delta =0.38346841237225538981.
\end{align*}
Now, using the above values of $\ga$, $\gb$, $\gg$, and $\gd$ we obtain the position vectors $\tilde p$, $\tilde q$, $\tilde r$ and $\tilde s$ as follows:
\begin{align*}
\tilde p&=(\frac{1}{2},0.5436907490155839431), \\
\tilde q&=(\frac{1}{2},0.1926448341274137497 ) , \\
\tilde r&= ( 0.2302330149367283460,0.1649562245075873150 ) , \\
\tilde s&=( 0.769766985063271654,0.1649562245075873150 ).
\end{align*}
Hence, the points $(\frac{1}{2},0.5436907490155839431), \, (\frac{1}{2},0.1926448341274137497), \\ (0.2302330149367283460, 0.1649562245075873150)$ \\
and $(0.769766985063271654,0.1649562245075873150)$  form an optimal set of four-means. Notice  that due to symmetry there are three optimal sets of four-means. As before, we can also calculate the quantization error in this case.

\begin{figure}
\centerline{\includegraphics[width=300pt]{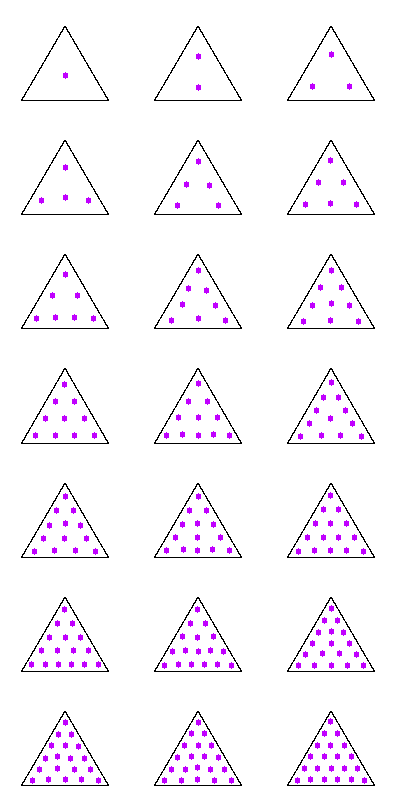}}
\caption{Results of a numerical search algorithm for $1\leq n\leq 21$ points, rotated so that the symmetry axis is vertical.\label{f:pm}}
\end{figure}

\section{Optimal sets of $n$-means}

As the number of points increases, so does the number of algebraic equations to be solved.  We apply a numerical search algorithm that makes random shifts to the point locations,
accepting better configurations, and gradually decreasing the shift amplitude in the absence of improvement.  In Figure~\ref{f:pm} we present the results of this numerical search
for $n\leq 21$ points.  Based on these results we make the following conjectures (``most'' means a set with density greater than 1/2):

\begin{conj}
For most $n$, there is an optimal configuration with at least one line of symmetry.
\end{conj}
In Figure~\ref{f:pm} this line of symmetry is chosen to be vertical.  In each case the number of points on each side of the vertical line is equal, however for $n=8$ and $n=19$,
the locations of points do not appear to be quite symmetrical.

We also note that when $n$ is a triangular  number, the points lie very close to a triangular lattice, and for other values, are located in identifiable rows, and are close to the
union of two subsets of triangular lattices.  Specifically

\begin{conj}
For most $n$, there is an optimal configuration with $N=\lfloor\sqrt{2n}\rfloor$ rows.  The $j$th row has $j$ points for $j\leq J$ where $J=N-|n-N(N+1)/2|$.  If
$n>N(N+1)/2$ the rows with $j>J$ each have one extra point (so, the jth row has $j+1$ points), while if $n<N(N+1)/2$ they each have one fewer point (so, the
$j$th row has $j-1$ points).
\end{conj}

Notice that $\lfloor\sqrt{2n}\rfloor$ identifies the closest triangular  number to a natural number $n$.  The conjecture is not stated for all $n$ as possible exceptions are $n=12$
(wrong number of rows) and $n=14$ (wrong distribution of points in rows).

When $n$ is a triangular  number $N(N+1)/2$, the locations are close to a triangular lattice, and it is possible to obtain a good bound on the quantization error:
\begin{theorem}\label{th:6.3}
When $n=N(N+1)/2$ for some positive integer $N\geq 3$, the quantization error is controlled by the bound
\begin{equation*}
V_n\leq \frac{45N^3-28\sqrt{21}N^2+(301-28\sqrt{21})N-98}{324N^3(N-1)^2}=\frac{5}{36 N^2}-\frac{14\sqrt{21}-45}{162 N^3}+O(N^{-4}).
\end{equation*}

\end{theorem}

\begin{figure}
\centerline{\includegraphics[width=400pt]{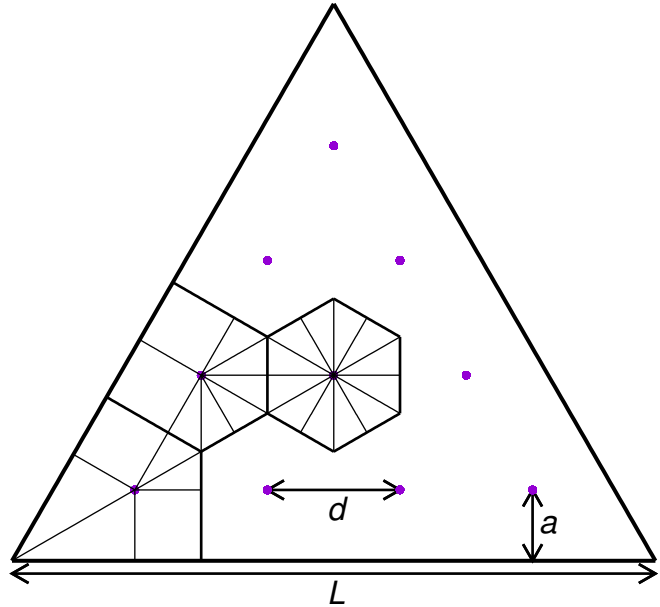}}
\caption{\label{f:5} The construction in the proof of  Theorem~\protect\ref{th:6.3}, illustrating the split of the Voronoi regions of centre, edge and corner points
into triangles and rectangles.  The value of $a$ chosen is the optimal value $a_{\mbox{opt}}$ defined below.}
\end{figure}

\begin{proof}
The proof is by direct calculation for the specific configuration shown in Figure~\ref{f:5}.  The points lie on a triangular lattice aligned with the triangular domain and have Voronoi regions as shown.  There are two parameters, the lattice spacing $d$, and the distance from any of the edge or corner
points to the edge of the triangle $a$.  We set $L$ to be the side length of the large triangle (set equal to unity at the end), so that the area is
$\mbox{Area}=L^2\sqrt{3}/4$. We then have
\begin{equation*}
L=(N-1)d+2\sqrt{3}a.
\end{equation*}
It is convenient to make $d$ the subject of this equation and substitute into the expressions below. Placing a point at the origin, we can find the quantization error
due to right triangular or rectangular domains:
\begin{align*}  V_{\pi/6}(r)&=\int_0^r dx \int_0^{x/\sqrt{3}} dy \frac{x^2+y^2}{\mbox{Area}}=\frac{10 r^4}{27 L^2},\\
V_{\pi/3}(r)&=\int_0^r dx \int_0^{x\sqrt{3}} dy \frac{x^2+y^2}{\mbox{Area}}=\frac{2r^4}{L^2},\\
V_{\mbox{rect}}(l,w)&=\int_0^l dx \int_0^w dy \frac{x^2+y^2}{\mbox{Area}}=\frac{4lw(l^2+w^2)}{3\sqrt{3}L^2}.
\end{align*}
Then, each point has a combination of these contributions
$V_{\mbox{center}}=12 V_{\pi/6}(d/2), \, V_{\mbox{edge}}=6V_{\pi/6}(d/2)+2V_{\mbox{rect}}(d/2,a), \, V_{\mbox{corner}}=2V_{\pi/6}(d/2)+2V_{\mbox{rect}}(d/2,a)+2V_{\pi/3}(a),$
and the overall quantization error (giving a bound for the optimal quantization error) is a sum of these, counting the number of points of each type
\begin{eqnarray}
V_n&\leq&\frac{(N-3)(N-2)}{2}V_{\mbox{center}}+3(N-2)V_{\mbox{edge}}+3V_{\mbox{corner}}\nonumber\\
&=&\frac{144\sqrt{3}a^4N(N-2)+144a^3N(N-2)L+144\sqrt{3}a^2L^2-84aL^3+5\sqrt{3}L^4}{144(N-1)^2}.\label{e:Vn}
\end{eqnarray}
Expanding for large $N$ and $L$, keeping both quantities at the same order, gives to leading order the optimal
\[ a_{\mbox{opt}}=\frac{\sqrt{7}L}{6N} \]
which, substituted into the expression (\ref{e:Vn}) gives the stated result.
\end{proof}

In the general case (arbitrary $n$) we have an asymptotic result:
\begin{cor} The quantization error satisfies
\[ V_n\leq  \frac{5}{72n} +O(n^{-3/2}) \]
as $n\to\infty$.
\end{cor}
\begin{proof}
This follows from Theorem~\ref{th:6.3}. For arbitary $n$, the distance to the previous triangular number is order $\sqrt{n}$. Thus we can add the extra points without increasing the leading term of the quantization error.
\end{proof}

We expect that the triangular lattice is optimal to leading order, so that $\leq$ may be replaced by $\sim$.
Furthermore, by placing a triangular lattice within a more general domain, we expect
\begin{conj}
If we consider a measure $P$ uniform on a domain with finite area $A$ and finite perimeter, then as $n\to\infty$,
\[ V_n \sim \frac{5\sqrt{3}A}{54n}. \]
\end{conj}


\begin{thebibliography}{9999}



\bibitem[DFG]{DFG} Q. Du, V. Faber and M. Gunzburger, \emph{Centroidal Voronoi Tessellations: Applications and Algorithms}, SIAM Review, Vol. 41, No. 4 (1999), pp. 637-676.


\bibitem[GG]{GG} A. Gersho and R.M. Gray, \emph{Vector quantization and signal compression}, Kluwer Academy publishers: Boston, 1992.

\bibitem[GKL]{GKL}  R.M. Gray, J.C. Kieffer and Y. Linde, \emph{Locally optimal block quantizer design}, Information and Control, 45 (1980), pp. 178-198.


\bibitem[GN]{GN}  R. Gray and D. Neuhoff, \emph{Quantization,} IEEE Trans. Inform. Theory,  44 (1998), pp. 2325-2383.


\bibitem[H]{H} J. Hutchinson, \emph{Fractals and self-similarity}, Indiana Univ. J., 30 (1981), 713-747.


\bibitem[GL1]{GL1} S. Graf and H. Luschgy, \emph{Foundations of quantization for probability distributions}, Lecture Notes in Mathematics 1730, Springer, Berlin, 2000.


\bibitem[GL2]{GL2} S. Graf and H. Luschgy, \emph{The Quantization of the Cantor Distribution}, Math. Nachr., 183 (1997), 113-133.



\bibitem[Z]{Z} R. Zam, \emph{Lattice Coding for Signals and Networks: A Structured Coding Approach to Quantization, Modulation, and Multiuser Information Theory}, Cambridge University Press, 2014.


\end{thebibliography}
\end{document}